\documentclass{cccg23} 
\usepackage[T1]{fontenc}
\usepackage[utf8]{inputenc}

\usepackage{amsfonts,amsmath,amssymb}
\usepackage{cite,microtype,graphicx}
\newtheorem{observation}[theorem]{Observation}

\theoremstyle{definition}

\title{A Parameterized Algorithm for Flat Folding}
\author{David Eppstein\thanks{Department of Computer Science, University of California, Irvine. Research supported in part by NSF grant CCF-2212129.}}

\date{ }

\begin{document}
\maketitle  

\begin{abstract}
We prove that testing the flat foldability of an origami crease pattern (either labeled with mountain and valley folds, or unlabeled) is fixed-parameter tractable when parameterized by the ply of the flat-folded state and by the treewidth of an associated planar graph, the cell adjacency graph of an arrangement of polygons formed by the flat-folded state. For flat foldings of bounded ply, our algorithm is single-exponential in the treewidth; this dependence on treewidth is necessary under the exponential time hypothesis.
\end{abstract}

\section{Introduction}

In a foundational result in the computational complexity of mathematical paper folding, Bern and Hayes proved in 1996 that it is $\mathsf{NP}$-complete to determine whether a crease pattern, described as a set of straight fold lines on a flat piece of paper, can be folded to lie flat again after exactly the prescribed folds have been made~\cite{BerHay-SODA-96}. This result holds regardless of whether the folds are given purely as line segments, or whether they additionally specify whether each fold is to be a mountain fold or a valley fold. It assumes a general model of folding where only the existence of the desired folded state is to be determined, and not a sequence of motions that reach it, but subsequent work has also proved similar hardness results for other models such as \emph{box pleating}, where the folds are aligned with the axes and diagonals of a square grid~\cite{AkiCheDem-JCGCGG-15}, and the \emph{simple folding} typical of sheet-metal manufacturing in which this motion must only be made on one fold line at a time~\cite{AkiDemKu-JIS-17,ArkBenDem-CGTA-04}.

On the positive side, not much is known about classes of crease patterns for which foldability is easier to determine. One such class, but a very limited one, is the class of patterns where the folds meet in a single vertex (or as a degenerate case, where they all lie on parallel lines). In this case, a linear-time greedy algorithm follows from the big-little-big lemma, in which creases forming a sharp angle between two wider angles must fold in a fixed way, allowing a reduction to a simpler configuration~\cite{BerHay-SODA-96}. Two more polynomial cases are simple folding of rectangles subdivided into congruent rectangles (``map folding'')~\cite{ArkBenDem-CGTA-04}, and general map folding of $2\times n$ grids of rectangles~\cite{Mor-12}.

In this work, we provide the first algorithmic upper bounds on testing flat foldability of arbitrary crease patterns, not restricted to special cases such as map folding. Our work analyzes this problem using tools from parameterized complexity. We show that flat-foldability is \emph{fixed-parameter tractable} when parameterized by two values: the \emph{ply} of the crease pattern (how many layers of paper can overlap at any point of the flat-folded result), and the \emph{treewidth} of an associated \emph{cell adjacency graph} constructed by overlaying the flat polygons of the crease pattern in the positions they would take in their folded state. The pattern may either be labeled with mountain and valley folds or unlabeled. We identify a wide class of patterns for which flat foldability is easy: those with bounded ply and bounded treewidth. For flat foldings of bounded ply, our algorithm is single-exponential in the treewidth. As we show in an appendix, this exponential dependence is necessary under the exponential time hypothesis, both for unlabeled and labeled crease patterns. We do not have as strong an argument for why the dependence on ply is necessary, but if it could be eliminated, we could solve map folding in polynomial time, a major open problem in this area.

Bounded ply is natural in paper folding, as large ply can lead to difficulty in the physical realization of a folding~\cite{DemEppHes-JDA-16}. The treewidth parameter is intended to capture the notion of a crease pattern that is complicated only in one dimension, and simple in a perpendicular dimension, as occurs (with large ply) for $2\times n$ map folding.  Single-vertex crease patterns also automatically have low treewidth (their cell adjacency graph is just a cycle; see \cref{sec:cell-adj}) but may again have high ply. Fixed-parameter tractability of an algorithm means that its worst-case time bound has the form of a polynomial in the input size, multiplied by a non-polynomial function of the parameters; in our case this function is factorial in the ply and exponential in the treewidth. On inputs for which the parameters are bounded, this function value is also bounded and the time bound simplifies to being purely a polynomial of the input size.

\begin{figure*}[t]
\includegraphics[width=\textwidth]{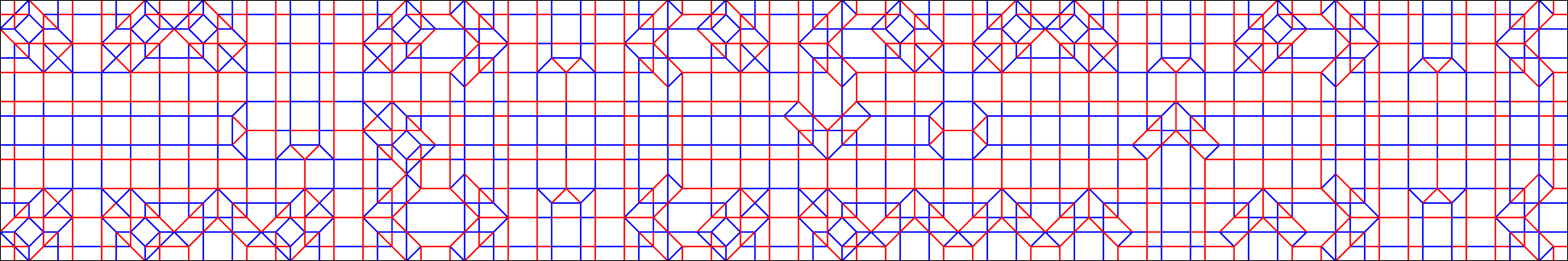}
\caption{A crease pattern for the origami font of Demaine, Demaine, and Ku, produced by \url{http://erikdemaine.org/fonts/maze/?text=origami}.}
\label{fig:origami-maze-cp}
\end{figure*}

\begin{figure}[t]
\centering\includegraphics[width=0.8\columnwidth]{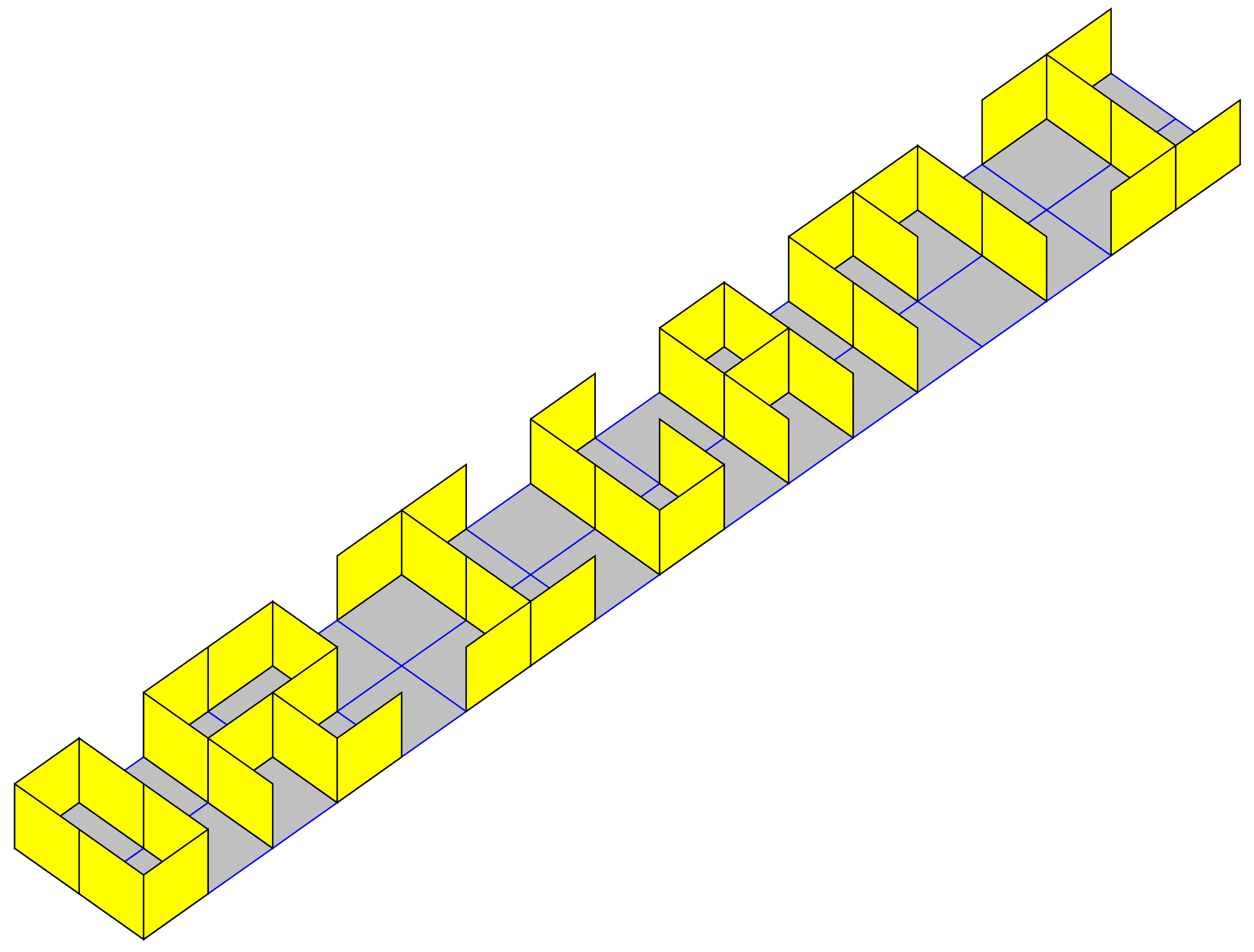}
\caption{The 3d folded form of the pattern from \cref{fig:origami-maze-cp}, as produced by \url{http://erikdemaine.org/fonts/maze/?text=origami}.}
\label{fig:origami-maze-3d}
\end{figure}

Another class of example patterns for which the parameters of our algorithm are naturally bounded comes from the origami font of Demaine, Demaine, and Ku~\cite{DemDemKu-G4G-10,DemDem-TCS-15,DemDem-JIP-20}. Rendering text in this font converts it into an origami crease pattern (\cref{fig:origami-maze-cp}). When folded, this pattern produces a three-dimensional structure consisting of letterform-shaped vertical walls on a flat background surface (\cref{fig:origami-maze-3d}). The resulting structures are not actually flat foldings (because of the vertical walls) but can easily be modified to be. The resulting crease pattern, for a line of text, has bounded ply, high complexity along any horizontal line through the pattern, and low complexity along any vertical line. Its cell adjacency graph has bounded bandwidth, but for a modified version of the font that included ascenders and descenders it would instead have bounded pathwidth, both of which are special cases of our bounded treewidth assumption.

\section{Preliminaries}

\subsection{Flat folding}

Following our previous work~\cite{Epp-JoCG-19}, we base our definition of flat folding on a \emph{local flat folding}, a simplified model of folding which describes only how the folding maps a flat surface to itself, and does not describe the spatial arrangement of the layers of paper as a flat-folded surface. We will then augment this model to include layer ordering, to define a \emph{flat folding}.

Thus, we define a \emph{local flat folding} of a planar polygon $P$ 
 to be a \emph{continuous piecewise isometry} $\varphi$ from $P$ to the plane. That is, it is a continuous function that acts as a distance-preserving mapping of the plane within each of a system of finitely many interior-disjoint polygons whose union is $P$.
 The points at which $\varphi$ is not locally an isometry lie on the boundaries of these polygons, forming \emph{creases} (line segments between two polygons mapped differently by $\varphi$) and \emph{vertices} (points where multiple creases meet). We may choose the polygons of $\varphi$ so that each polygon is bounded by creases and by the boundary of $P$.  The \emph{crease pattern} of a local flat folding is this system of creases and vertices. At this level of detail, there is no distinction between mountain folds and valley folds.

\begin{observation}
Given a decomposition of a polygon $P$ into smaller polygons, we can determine in linear time whether this decomposition forms the crease pattern of a local flat folding, and if so reconstruct a function $\varphi$ having that decomposition as its crease pattern.
\end{observation}

\begin{proof}
We choose an arbitrary starting polygon, set $\varphi$ to be the identity within this polygon, and then traverse the adjacencies between polygons of the decomposition. When we traverse the edge between a polygon whose mapping under $\varphi$ has been determined to another polygon whose mapping has not, we set the mapping for the new polygon to be the mapping for the old polygon, reflected across the line through the traversed edge. When we traverse an edge to a polygon whose mapping has already been determined, we check that its mapping is consistent with this reflection.
\end{proof}

The function $\varphi$, constructed in this way, is unique up to rigid transformations of the plane.

\begin{figure}[t]
\includegraphics[width=\columnwidth]{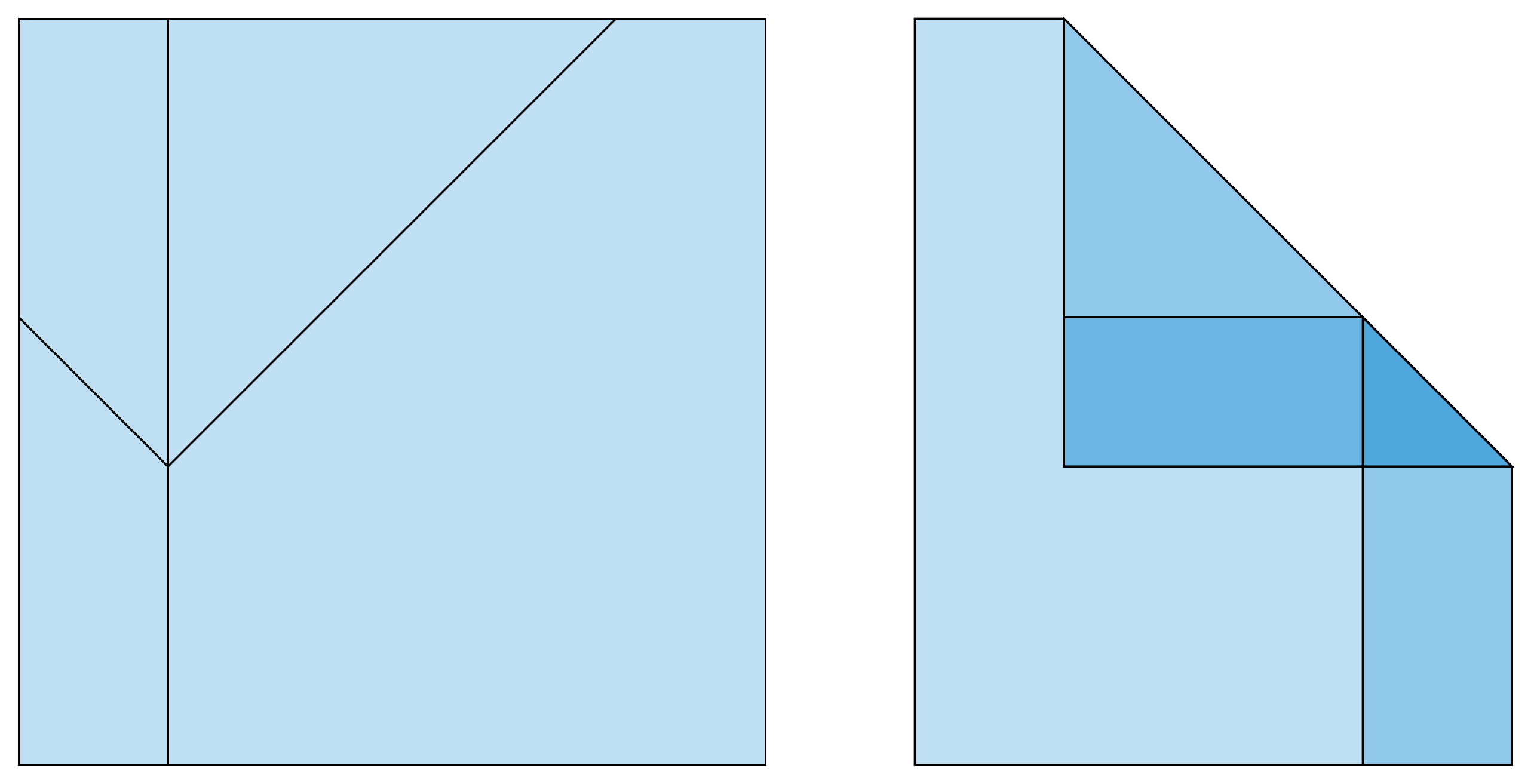}
\caption{The crease pattern of a local flat folding (left) and the arrangement of the folding (right), with shading indicating the ply of each arrangement cell. The ply of the overall pattern is four, equal to the maximum ply in the small triangular cell.}
\label{fig:arrangement-ply}
\end{figure}

We define the \emph{arrangement} of a local flat folding to be the result of overlaying the transformed copies of each of its polygons. It partitions the plane into \emph{cells}, polygons that are not crossed by the image of any crease. Within each cell, all points have preimages coming from the same set of polygons of the crease pattern. The \emph{ply} of a cell is the number of these preimages, and the ply of the crease pattern is the maximum ply of any cell. See \cref{fig:arrangement-ply}. Using standard methods from computational geometry, an arrangement of a local flat folding with $n$ creases has complexity $O(n^2)$ and can be constructed (including the calculation of its ply) in time $O(n^2)$.

\begin{figure}[t]
\centering\includegraphics[width=0.8\columnwidth]{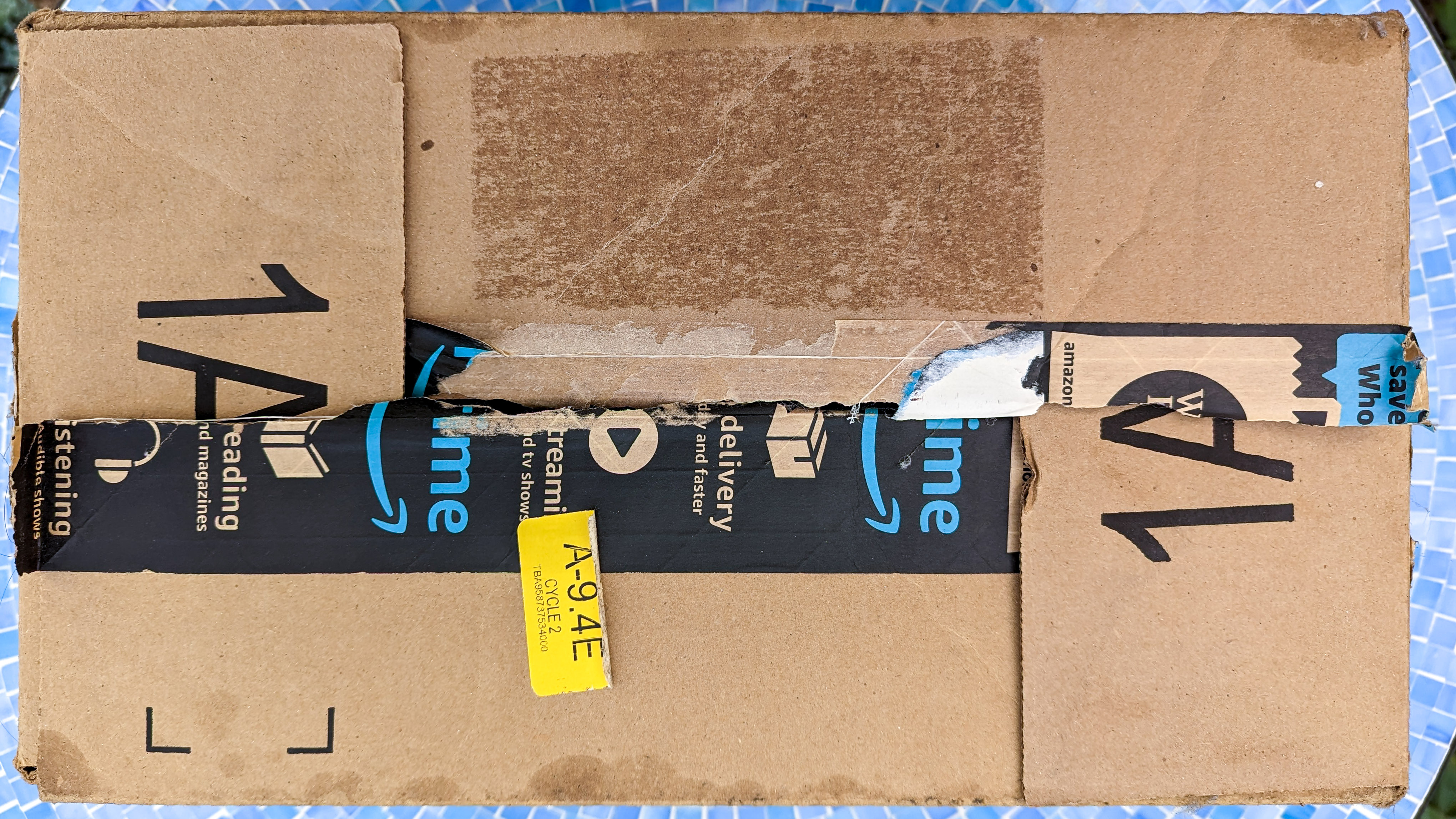}
\caption{Cyclically-ordered box-top flaps}
\label{fig:boxtops}
\end{figure}

\begin{figure}[t]
\centering\includegraphics[width=\columnwidth]{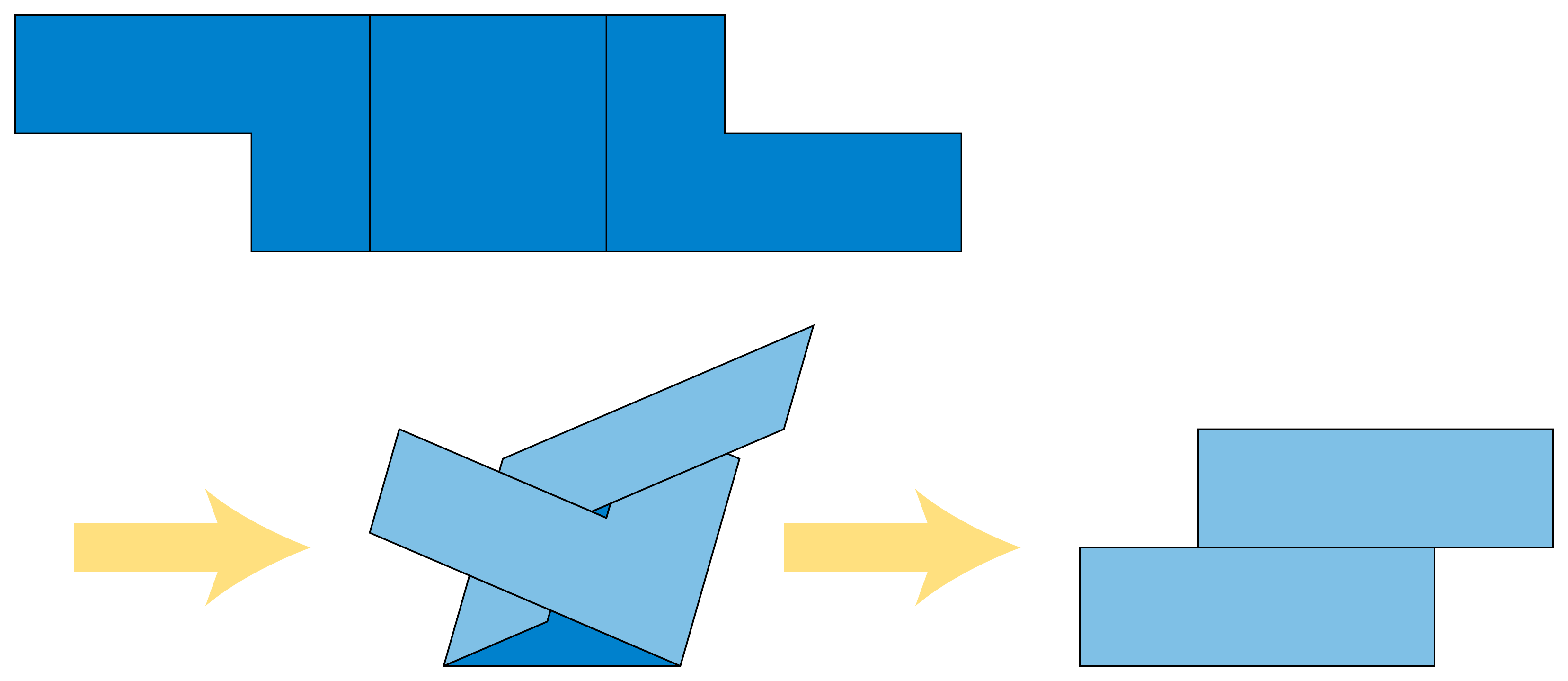}
\caption{A crease pattern with two valley folds that, when flat-folded, causes its two L-shaped polygons to have two different above-below orderings in the two cells of the arrangement where they overlap.}
\label{fig:over-under}
\end{figure}

Our previous work~\cite{Epp-JoCG-19} defined a \emph{global flat folding} to be ``a local flat folding that, for every $\varepsilon>0$, is $\varepsilon$-close to a topological embedding of the plane into three-dimensional space'', but for our purposes we need to actually describe the three-dimensional embedding combinatorially, not merely to assert its existence. Instead, we define a \emph{layering} of a local flat folding to be an assignment, for each cell of the arrangement of the folding, of a vertical ordering on the preimage polygons of the cell. We allow different cells to have different and inconsistent vertical orderings. This may be necessary to model real-world foldings in which the vertical ordering of polygon has cycles, as happens for instance in flexagons~\cite{Gar-SA-56} and in a common method for folding the four flaps of a box top (\cref{fig:boxtops}). It is even possible for the same two polygons of a crease pattern to have two different above-below orderings in two different cells of the arrangement in which they overlap (\cref{fig:over-under}).

We define a \emph{flat folding} to be a local flat folding together with a layering that, for every $\varepsilon>0$, is consistent with the layering coming from a topological embedding of the crease pattern into three-dimensional space that is $\varepsilon$-close to the local flat folding. Here, ``close'' means there exists a local flat folding into a plane in space so that, for every point of the crease pattern, its images under the topological embedding and under the local flat folding have distance at most $\varepsilon$ from each other. To avoid topological difficulties we additionally require that a line perpendicular to the plane, through a point of the plane farther than $\varepsilon$ from any crease, has exactly one point of intersection with each polygon in the topological embedding: the embedding cannot be ``crumpled'' far from its creases. With this restriction, the polygons that map to each cell have a consistent layering, the ordering in which they meet any such perpendicular line.

If we look at a cross-section of such a topological embedding, across any crease of the embedding, we will see the layers in two adjacent cells of the arrangement. Two layers in the same cell can be paired up to form a crease, two layers from the two cells can be paired up to form parts of a polygon that span the cell without forming a crease, and it is also possible to have an unpaired layer whose boundary at the crease coincides with a boundary of the overall crease pattern (\cref{fig:layer-consistency}, left). These layers and pairs of layers must meet certain obvious conditions:
\begin{itemize}
\item If two polygons span the two cells without being creased, they must be consistently ordered in both cells instead of crossing at the crease (\cref{fig:layer-consistency}, top right).
\item If two layers of the same cell meet in a crease, and another polygon spans the two cells without being creased, the polygon cannot lie between the two creased layers of the first polygon, as their crease would block it from extending into the second cell  (\cref{fig:layer-consistency}, middle right).
\item If two pairs of layers in the same cell meet in the same crease, then their layers cannot alternate, as this would again form a crossing (\cref{fig:layer-consistency}, bottom right). However, it may be possible to have alternating pairs of layers that meet in different creases, along different edges of the same cell.
\item If two layers of the same cell meet in a crease, and are labeled as being a mountain fold or valley fold in the crease pattern, then the ordering of the layers must be consistent with that type of fold (not shown).
\end{itemize}
We define a layering for a local flat folding to be \emph{uncrossed} when, at each crease, it meets all of these conditions.

\begin{lemma}
A local flat folding comes from a flat folding if and only if it has an uncrossed layering.
\end{lemma}

\begin{proof}
In one direction, if a flat folding exists, it cannot violate any of the conditions above, because each describes a certain type of crossing, and topological embeddings forbid crossings. In the other direction, every uncrossed layering comes from a flat folding: one can form a 3d embedding from it, by shrinking each cell a small distance from its boundary, making parallel copies of the cell in 3d in the order given by the layering, all separated from each other but within distance $\varepsilon$ of the plane of the local flat folding, and connecting them with curved patches of surface near each crease.

\begin{figure}[t]
\includegraphics[width=\columnwidth]{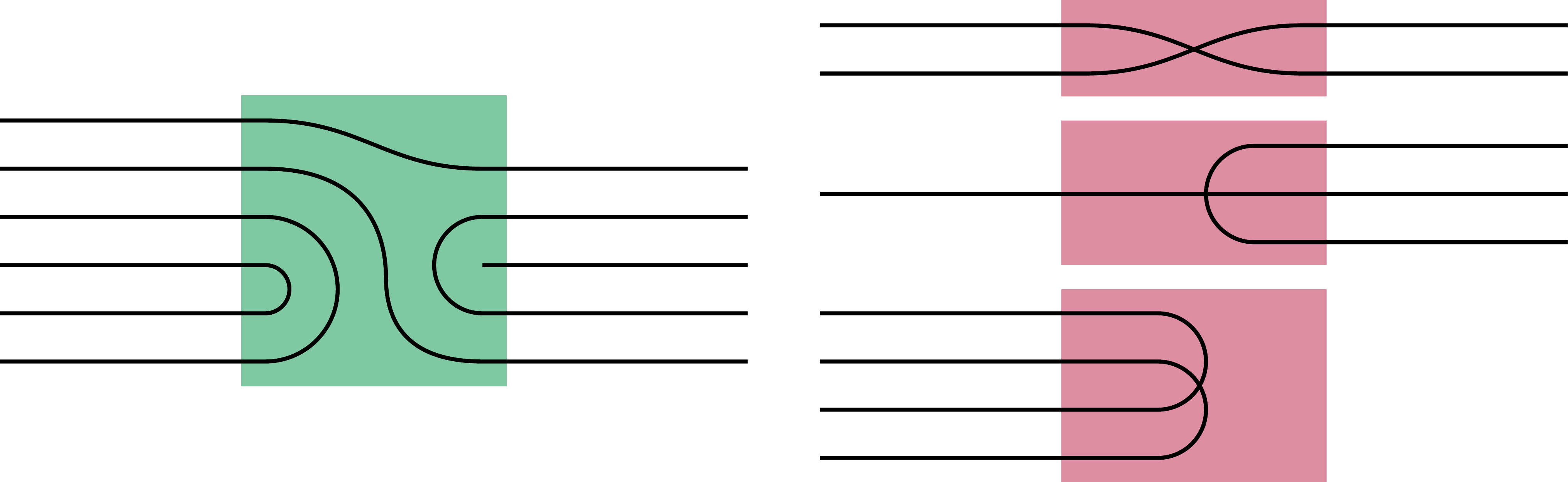}
\caption{Left: cross-section through a crease (shaded region) of a uncrossed layering. Right: Three ways that a layering can be inconsistent across a crease: two uncreased polygons cross (top), an uncreased polygon is blocked by two layers that connect to form a crease (middle), or two pairs of creased layers cross (bottom).}
\label{fig:layer-consistency}
\end{figure}

It is unnecessary to add more case analysis for the way layerings can interact at a vertex, instead of across a crease. Two surfaces in 3d cannot cross each other at a single point, without crossing along a curve touching that point, so if a system of surfaces in 3d defined from a uncrossed layering avoids crossings except at points $\varepsilon$-close to the vertices, it can be converted into a topological embedding for the same layering that avoids crossing everywhere.
\end{proof}

\subsection{Treewidth}

A \emph{tree decomposition} of a graph $G$ consists of an unrooted tree $T$, and an assignment to each tree vertex $t_i$ of a set $B_i$ of vertices from $G$ (called a \emph{bag}), such that each vertex of $G$ belongs to the bags from a connected subtree of $T$, and each edge of $G$ has endpoints that belong together in at least one bag. Its \emph{width} is the maximum size of a bag, minus one, and the \emph{treewidth} of $G$ is the minimum width of any tree decomposition of $G$. Many optimization problems that are hard on arbitrary graphs can be solved in linear time on graphs of bounded treewidth, using dynamic programming over their tree decompositions. Although finding the treewidth is itself a hard optimization problem, it can be solved in linear time for graphs of bounded treewidth, with a time bound that is exponential in the cube of the width~\cite{Bod-SICOMP-96}. In our application we will be using the treewidth of planar graphs, derived from the arrangement of a crease pattern. It is unknown whether planar treewidth is hard, but it can be approximated in (unparameterized) polynomial time with an approximation ratio of $3/2$ by an algorithm for a closely related width parameter called \emph{branchwidth}~\cite{SeyTho-Comb-94}.

It will simplify the description of our algorithm to use a tree decomposition of a special form, called a \emph{nice tree decomposition}. This differs from a tree decomposition in being a rooted tree. The tree vertices and their bags have four types:
\begin{itemize}
\item \emph{Leaf bags}, leaves of the rooted tree, have exactly one graph vertex in the bag.
\item \emph{Introduce bags} have exactly one child vertex in the tree, and their bag differs from that of the child by the addition of exactly one graph vertex.
\item \emph{Forget bags} have exactly one child vertex in the tree, and their bag differs from that of the child by the removal of exactly one graph vertex.
\item \emph{Join bags} have exactly two children, whose bags are both equal to the join bag.
\end{itemize}

A nice tree decomposition can be constructed in linear time from an arbitrary tree decomposition, without increasing the width, and it has size linear in the size of the input tree decomposition~\cite{Klo-94}.

\subsection{Cell adjacency graphs and their treewidth}
\label{sec:cell-adj}

\begin{figure}[t]
\centering\includegraphics[scale=0.3]{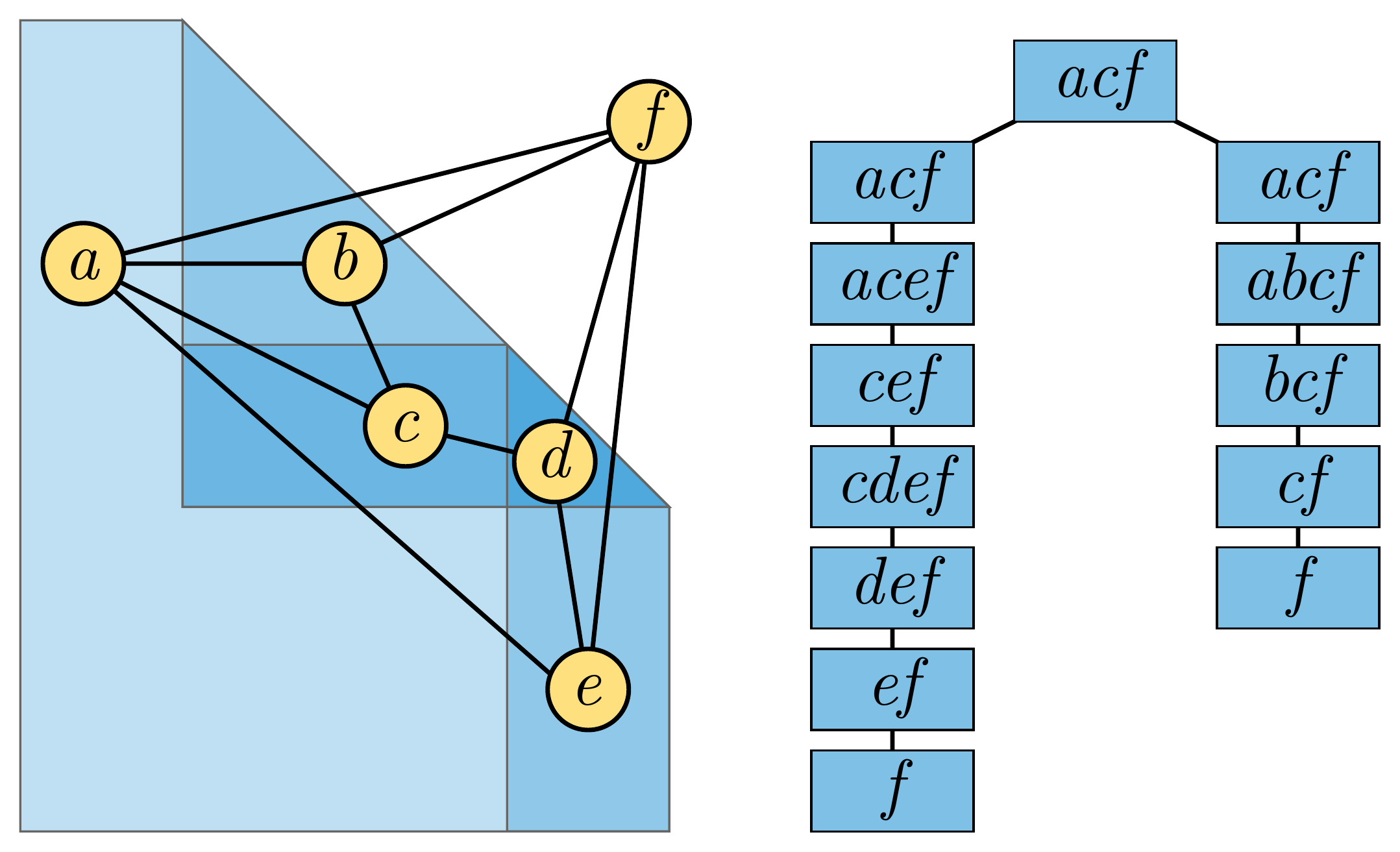}
\caption{The arrangement from \cref{fig:arrangement-ply}, its cell adjacency graph, and a nice tree decomposition of the cell adjacency graph.}
\label{fig:cell-adj-decomp}
\end{figure}

Recall that our definition of flat folding involves constructing an arrangement of polygons, the images of the polygons in the crease pattern under the mapping that defines a local flat folding. The usual notion of an \emph{arrangement graph} is a planar graph with a vertex for each crossing or endpoint of a line segment in this arrangement, and an edge for each piece of polygon boundary connecting two of these vertices~\cite{BosEveWis-IJCGA-03}. Instead, we use its dual graph, which we call the \emph{cell adjacency graph}. This has a vertex for each cell of the arrangement, and an edge between each two neighboring cells. It has been used before in computational geometry (e.g.~\cite{DenShu-DCG-88}), but appears to lack a standard name.

Even when a cell has ply zero, we include it in this graph, in order to check for crossings along the creases between this cell and its neighbors. For example, in the map folding problem, a square grid crease pattern is folded down to a single square, but the arrangement has two cells, the inside of the square and the outside, so the cell adjacency graph is $K_2$. In the case of a single-vertex crease pattern, the local flat folding produces an arrangement consisting of wedges all having this vertex as their apex, and its cell adjacency graph is a cycle.

The two main parameters for the analysis of our algorithm will be the ply of the local flat folding, and the treewidth of the cell adjacency graph. \cref{fig:cell-adj-decomp} depicts an example of a cell adjacency graph of treewidth 2, and a nice tree decomposition with a join bag at its root. 

\section{The algorithm}

We will test the flat foldability of a crease pattern by first attempting to construct its local flat folding. If this step fails, a flat folding does not exist, and our algorithm exits with a negative answer. Next, we construct its arrangement and its cell adjacency graph,  find an optimal or near-optimal tree decomposition of the cell adjacency graph using any of the various algorithms known for this problem, and convert the tree decomposition to a nice tree decomposition of the same width.

Finally, we reach the main part of our algorithm: a bottom-up dynamic program on the bags of the tree decomposition.
If $B$ is any bag (that is, a set of cells of the arrangement, associated with a vertex of the nice tree decomposition),
we define a \emph{state} of $B$ to be a layering of each cell in $B$.

\begin{observation}
In a tree decomposition of width $w$ for a crease pattern of ply $p$, every bag has at most $(p!)^{w+1}$ states.
\end{observation}

If $B$ has a child $C$ in the tree decomposition,
then we say that a state of $B$ is \emph{consistent} with a state of $C$ if they have the same layering in all of the cells that belong to both bags. We say that a state of bag $B$ is \emph{locally uncrossed} if, for all pairs of adjacent cells that both belong to $B$, the layerings of these two cells in this state meet the same conditions that we used earlier to define a global layering as being uncrossed. We say that a state is \emph{valid} when it is locally uncrossed and is consistent with (recursively defined) valid states for all child bags.

\begin{lemma}
\label{lem:valid-fold}
For any bag $B$ of the tree decomposition, there exists a valid state for $B$ if and only if there exists a layering for the entire local flat folding that meets the conditions of being uncrossed at all creases between pairs of cells that occur together in $B$ or one of its descendants in the tree decomposition.
\end{lemma}

\begin{proof}
If such a layering exists, its restriction to the cells in $B$ and its descendant bags produces a valid state.
If a valid state exists, coming from a recursively constructed set of valid states among its descendant bags, then each of these states must consistently layer the cells that they have in common, by the requirement of tree-decompositions that each graph vertex belong to bags in a connected subtree. Form a global layering by choosing arbitrarily a layering for each cell that is not included among these descendants. Then it must be uncrossed at all creases between pairs of cells that occur together in $B$ or one of its descendants, because any crossing would cause the state to be invalid at that bag, violating the assumption that we have a recursively constructed set of valid states.
\end{proof}

\begin{lemma}
\label{lem:compute-valid}
If we have already computed the valid states of each child of a given bag $B$ of a nice tree decomposition, we can compute the valid states for $B$ itself in time $O(pw(p!)^{w+1})$.
\end{lemma}

\begin{proof} 
We apply a case analysis according to the type of $B$ in the decomposition.
\begin{itemize}
\item At a leaf bag, all states are valid, because there are no creases between pairs of cells to cause crossings.
\item At an introduce bag, we must add a layering for the introduced cell to all valid layerings of the other cells from the child node. For each child layering, and each layering of the introduced cell, we check at most $w$ previously-unrepresented creases, each in time $O(p)$, to determine whether it forms any of the forbidden crossing types.
\item At a forget bag, all valid states of the child node determine a valid state of the bag, by forgetting the layering on the cell that is not included.
\item At a join bag, a state is valid when it is valid in both child states. We can intersect the sets of valid states in both children, in time linear in the number of possible states, using a bit array.\qedhere
\end{itemize}
\end{proof}

Putting these pieces together gives our main result:

\begin{theorem}
Testing flat foldability of a crease pattern with $n$ creases and ply $p$, with a cell adjacency graph of treewidth $w$, can be performed in time that is fixed-parameter tractable in $p$ and $w$, and quadratic in $n$.
\end{theorem}

\begin{proof}
We construct the nice tree decomposition as described above, and traverse it in bottom-to-top order, using \cref{lem:compute-valid} to determine the valid states in each bag. A folding exists if and only if there is a valid state at the root bag, by  \cref{lem:valid-fold}.
The quadratic dependence on $n$ comes from the size of the arrangement of the local flat folding, and the size of the tree decomposition of its cell adjacency graph. The dependence on ply and width comes from the time bound per bag in \cref{lem:compute-valid}, the time to construct a tree decomposition using known algorithms, and the relation between the width of the cell adjacency graph and the width of the constructed decomposition coming from the choice of these algorithms.
\end{proof}

\section{Conclusions}

We have shown that flat foldability, in a general model allowing cyclic overlaps between polygons, can be tested in fixed-parameter tractable time when parameterized both by ply and by the treewidth of an associated cell adjacency graph. Both parameters appear necessary for this result: the known $\mathsf{NP}$-hardness reductions for flat foldability can be made to have bounded ply (but unbounded treewidth), while the still-open map folding problem has bounded treewidth and more strongly bounded cell adjacency graph size (but unbounded ply).

It would be of interest to extend our algorithms to other forms of flat folding, such as the \emph{simple folding} models~\cite{AkiDemKu-JIS-17,ArkBenDem-CGTA-04}.
Another direction for possible future work concerns models of folding that require the existence of a three-dimensional continuous motion respecting the given fold lines (rigid origami~\cite{StrWhi-JCDCG-05,PanStr-CG-10,AbeCanDem-JoCG-16}), as well as inputs where the desired folded state is in some way three-dimensional (such as the raised ridges in the origami fonts of Demaine, Demaine, and Ku~\cite{DemDemKu-G4G-10,DemDem-TCS-15,DemDem-JIP-20}. Although there has been extensive study of types of instance that can or cannot be guaranteed to have a continuous motion taking them between their unfolded and folded states~\cite{DemMit-CCCG-01,DemDevMit-CCCG-04,Epp-CCCG-22-luse}, there is little work on algorithmic time bounds for testing the existence of this sort of motion. Whether these three-dimensional models of origami can be reduced to a combinatorial problem to which the sort of methods described here can apply remains a challenge.

\bibliographystyle{plainurl}
\bibliography{folding}

\begin{thebibliography}{10}

\bibitem{AbeCanDem-JoCG-16}
Zachary Abel, Jason Cantarella, Erik~D. Demaine, David Eppstein, Thomas Hull,
  Jason~S. Ku, Robert~J. Lang, and Tomohiro Tachi.
\newblock {Rigid origami vertices: Conditions and forcing sets}.
\newblock {\em J. Comput. Geom.}, 7(1):171{--}184, 2016.
\newblock \href {http://dx.doi.org/10.20382/jocg.v7i1a9}
  {\path{doi:10.20382/jocg.v7i1a9}}.

\bibitem{AkiCheDem-JCGCGG-15}
Hugo~A. Akitaya, Kenneth~C. Cheung, Erik~D. Demaine, Takashi Horiyama, Thomas
  Hull, Jason~S. Ku, Tomohiro Tachi, and Ryuhei Uehara.
\newblock {Box pleating is hard}.
\newblock In Jin Akiyama, Hiro Ito, Toshinori Sakai, and Yushi Uno, editors,
  {\em Discrete and Computational Geometry and Graphs {--} 18th Japan
  Conference, JCDCGG 2015, Kyoto, Japan, September 14{--}16, 2015, Revised
  Selected Papers}, volume 9943 of {\em Lecture Notes in Comput. Sci.}, pages
  167{--}179. Springer, 2015.
\newblock \href {http://dx.doi.org/10.1007/978-3-319-48532-4_15}
  {\path{doi:10.1007/978-3-319-48532-4_15}}.

\bibitem{AkiDemKu-JIS-17}
Hugo~A. Akitaya, Erik~D. Demaine, and Jason~S. Ku.
\newblock {Simple folding is really hard}.
\newblock {\em J. Information Processing}, 25:580{--}589, 2017.
\newblock \href {http://dx.doi.org/10.2197/ipsjjip.25.580}
  {\path{doi:10.2197/ipsjjip.25.580}}.

\bibitem{ArkBenDem-CGTA-04}
Esther~M. Arkin, Michael~A. Bender, Erik~D. Demaine, Martin~L. Demaine, Joseph
  S.~B. Mitchell, Saurabh Sethia, and Steven~S. Skiena.
\newblock {When can you fold a map?}
\newblock {\em Comput. Geom. Theory {\&} Appl.}, 29(1):23{--}46, 2004.
\newblock \href {http://dx.doi.org/10.1016/j.comgeo.2004.03.012}
  {\path{doi:10.1016/j.comgeo.2004.03.012}}.

\bibitem{BerHay-SODA-96}
Marshall Bern and Barry Hayes.
\newblock {The complexity of flat origami}.
\newblock In {\em Proc. 7th ACM-SIAM Symposium on Discrete Algorithms (SODA
  '96)}, pages 175{--}183, Philadelphia, PA, 1996. Society for Industrial and
  Applied Mathematics.
\newblock URL: \url{https://portal.acm.org/citation.cfm?id=313852.313918}.

\bibitem{Bod-SICOMP-96}
Hans~L. Bodlaender.
\newblock {A linear-time algorithm for finding tree-decompositions of small
  treewidth}.
\newblock {\em SIAM J. Comput.}, 25(6):1305{--}1317, 1996.
\newblock \href {http://dx.doi.org/10.1137/S0097539793251219}
  {\path{doi:10.1137/S0097539793251219}}.

\bibitem{BosEveWis-IJCGA-03}
Prosenjit Bose, Hazel Everett, and Stephen Wismath.
\newblock {Properties of arrangement graphs}.
\newblock {\em Internat. J. Comput. Geom. Appl.}, 13(6):447{--}462, 2003.
\newblock \href {http://dx.doi.org/10.1142/S0218195903001281}
  {\path{doi:10.1142/S0218195903001281}}.

\bibitem{DemDem-TCS-15}
Erik~D. Demaine and Martin~L. Demaine.
\newblock {Fun with fonts: algorithmic typography}.
\newblock {\em Theor. Comput. Sci.}, 586:111{--}119, 2015.
\newblock \href {http://dx.doi.org/10.1016/j.tcs.2015.01.054}
  {\path{doi:10.1016/j.tcs.2015.01.054}}.

\bibitem{DemDem-JIP-20}
Erik~D. Demaine and Martin~L. Demaine.
\newblock {Adventures in maze folding art}.
\newblock {\em J. Information Processing}, 28:745{--}749, 2020.
\newblock \href {http://dx.doi.org/10.2197/ipsjjip.28.745}
  {\path{doi:10.2197/ipsjjip.28.745}}.

\bibitem{DemDemKu-G4G-10}
Erik~D. Demaine, Martin~L. Demaine, and Jason~S. Ku.
\newblock {Origami maze puzzle font}.
\newblock In {\em Exchange Book of the 9th Gathering for Gardner (G4G9)}. March
  24--28 2010.
\newblock URL: \url{https://erikdemaine.org/papers/MazeAlphabet_G4G9/}.

\bibitem{DemDevMit-CCCG-04}
Erik~D. Demaine, Satyan~L. Devadoss, Joseph S.~B. Mitchell, and Joseph
  O'Rourke.
\newblock {Continuous foldability of polygonal paper}.
\newblock In {\em Proceedings of the 16th Canadian Conference on Computational
  Geometry, CCCG'04, Concordia University, Montr{\'e}al, Qu{\'e}bec, Canada,
  August 9-11, 2004}, pages 64{--}67, 2004.
\newblock URL: \url{https://www.cccg.ca/proceedings/2004/55.pdf}.

\bibitem{DemEppHes-JDA-16}
Erik~D. Demaine, David Eppstein, Adam Hesterberg, Hiro Ito, Anna Lubiw, Ryuhei
  Uehara, and Yushi Uno.
\newblock {Folding a paper strip to minimize thickness}.
\newblock {\em J. Discrete Algorithms}, 36:18{--}26, 2016.
\newblock \href {http://dx.doi.org/10.1016/j.jda.2015.09.003}
  {\path{doi:10.1016/j.jda.2015.09.003}}.

\bibitem{DemMit-CCCG-01}
Erik~D. Demaine and Joseph S.~B. Mitchell.
\newblock {Reaching folded states of a rectangular piece of paper}.
\newblock In {\em Proceedings of the 13th Canadian Conference on Computational
  Geometry, University of Waterloo, Ontario, Canada, August 13-15, 2001}, pages
  73{--}75, 2001.
\newblock URL:
  \url{https://erikdemaine.org/papers/PaperReachability_CCCG2001/}.

\bibitem{DenShu-DCG-88}
Linda Deneen and Gary Shute.
\newblock {Polygonizations of point sets in the plane}.
\newblock {\em Discrete Comput. Geom.}, 3(1):77{--}87, 1988.
\newblock \href {http://dx.doi.org/10.1007/BF02187898}
  {\path{doi:10.1007/BF02187898}}.

\bibitem{Epp-JoCG-19}
David Eppstein.
\newblock {Realization and connectivity of the graphs of origami flat
  foldings}.
\newblock {\em J. Comput. Geom.}, 10(1):257{--}280, 2019.
\newblock \href {http://dx.doi.org/10.20382/jocg.v10i1a10}
  {\path{doi:10.20382/jocg.v10i1a10}}.

\bibitem{Epp-CCCG-22-luse}
David Eppstein.
\newblock {Locked and unlocked smooth embeddings of surfaces}.
\newblock In {\em Proc. 34th Canadian Conference on Computational Geometry
  (CCCG 2022)}, pages 135{--}142, 2022.

\bibitem{Gar-SA-56}
Martin Gardner.
\newblock {Flexagons: In which strips of paper are used to make hexagonal
  figures with unusual properties}.
\newblock {\em Scientific American}, 195(6):162{--}168, December 1956.
\newblock \href {http://dx.doi.org/10.1038/scientificamerican1256-162}
  {\path{doi:10.1038/scientificamerican1256-162}}.

\bibitem{ImpPatZan-JCSS-01}
Russell Impagliazzo, Ramamohan Paturi, and Francis Zane.
\newblock {Which problems have strongly exponential complexity?}
\newblock {\em J. Comput. System Sci.}, 63(4):512{--}530, 2001.
\newblock \href {http://dx.doi.org/10.1006/jcss.2001.1774}
  {\path{doi:10.1006/jcss.2001.1774}}.

\bibitem{JonLagNor-SODA-13}
Peter Jonsson, Victor Lagerkvist, Gustav Nordh, and Bruno Zanuttini.
\newblock {Complexity of SAT problems, clone theory and the exponential time
  hypothesis}.
\newblock In Sanjeev Khanna, editor, {\em Proceedings of the Twenty-Fourth
  Annual ACM-SIAM Symposium on Discrete Algorithms, SODA 2013, New Orleans,
  Louisiana, USA, January 6-8, 2013}, pages 1264{--}1277. SIAM, 2013.
\newblock \href {http://dx.doi.org/10.1137/1.9781611973105.92}
  {\path{doi:10.1137/1.9781611973105.92}}.

\bibitem{Klo-94}
Ton Kloks.
\newblock {\em {Treewidth, Computations and Approximations}}, volume 842 of
  {\em Lecture Notes in Comput. Sci.}
\newblock Springer, 1994.
\newblock \href {http://dx.doi.org/10.1007/BFb0045375}
  {\path{doi:10.1007/BFb0045375}}.

\bibitem{Mor-12}
Thomas~D. Morgan.
\newblock {Map Folding}.
\newblock Master's thesis, Massachusetts Institute of Technology, June 2012.
\newblock URL: \url{https://hdl.handle.net/1721.1/77030}.

\bibitem{PanStr-CG-10}
Gaiane Panina and Ileana Streinu.
\newblock {Flattening single-vertex origami: The non-expansive case}.
\newblock {\em Computational Geometry}, 43(8):678{--}687, 2010.
\newblock \href {http://arxiv.org/abs/1003.3490} {\path{arXiv:1003.3490}},
  \href {http://dx.doi.org/10.1016/j.comgeo.2010.04.002}
  {\path{doi:10.1016/j.comgeo.2010.04.002}}.

\bibitem{SeyTho-Comb-94}
Paul~D. Seymour and Robin Thomas.
\newblock {Call routing and the ratcatcher}.
\newblock {\em Combinatorica}, 14(2):217{--}241, 1994.
\newblock \href {http://dx.doi.org/10.1007/BF01215352}
  {\path{doi:10.1007/BF01215352}}.

\bibitem{StrWhi-JCDCG-05}
Ileana Streinu and Walter Whiteley.
\newblock {Single-vertex origami and spherical expansive motions}.
\newblock In {\em Discrete and Computational Geometry: Japanese Conference,
  JCDCG 2004, Tokyo, Japan, October 8-11, 2004, Revised Selected Papers},
  volume 3742 of {\em Lecture Notes in Comput. Sci.}, pages 161{--}173.
  Springer, 2005.
\newblock \href {http://dx.doi.org/10.1007/11589440_17}
  {\path{doi:10.1007/11589440_17}}.

\end{thebibliography}

\newpage
\appendix
\section{ETH-hardness}

In our parameterized algorithm for flat folding, the dependence on ply comes from \cref{lem:compute-valid}, which provides a time bound of $O(pw(p!)^{w+1})$ for computing the valid states of a single bag in a nice tree decomposition. The overall time bound is then this same bound, multiplied by the $O(n)$ bags of the decomposition. When $p=O(1)$, this bound reduces to single-exponential in~$w$: the total time is $O(n2^{O(w)})$.

 As we now show, a bound of this form is necessary under the exponential-time hypothesis~\cite{ImpPatZan-JCSS-01}, which for our purposes is most conveniently phrased as the assumption that there does not exist an algorithm for the 3SAT (satisfiability of 3-CNF Boolean formulae with $n$ variables and $m$ clauses) that has a sublinear running time bound of the form $2^{o(n+m)}$. Our proof uses NAE3SAT (not-all-equal-3-satisfiability), a variant of 3SAT in which there are again $n$ Boolean variables, and in which certain triples of variables and their negations are not allowed to be equal. Standard NP-completeness reductions from 3SAT to NAE3SAT produce instances with $O(n+m)$ variables and clauses, from which it follows that under the exponential time hypothesis it is not possible to solve NAE3SAT instances in time subexponential in their numbers of variables or clauses. The same is known to be true more generally for a wide class of satisfiability-like problems including both 3SAT and NAE3SAT~\cite{JonLagNor-SODA-13}.

We base our hardness result on the proof  by Bern and Hayes that flat foldability is NP-complete~\cite{BerHay-SODA-96}. Bern and Hayes actually provide two proofs, one for unlabeled crease patterns and one for crease patterns labeled with mountain folds and valley folds, but both follow the same outline. They are reductions from NAE3SAT, and they produce crease patterns in the shape of a rectangle, where each variable of a NAE3SAT instance is represented by two closely spaced parallel zigzag paths of creases from the left side of the rectangle to the right side; none of these paths cross each other. Each clause of the NAE3SAT instance is represented by a small folded area near the top of the rectangle. Pairs of closely-spaced vertical fold lines connect the clauses to the variables, passing through the zigzag paths of variables that they do not interact with. Additional ``noise'' pairs of closely-spaced vertical fold lines are necessary to produce the zigzag pattern of the variable creases, but otherwise pass through the other variables without interacting with them. Each variable path, and each vertical pair of fold lines, have two locally-consistent folded states (used in the proof to represent the true and false truth assignment to each variable). The clause regions can only be flat-folded for truth assignments that satisfy the given clause. When a flat folding exists, and the construction is flat-folded, most of the paper has ply 1, with ply 3 along the folded regions near each variable gadget and vertical fold line, ply 5 at the points where two of these folded regions cross, and somewhat larger ply within the clause gadgets. \cref{fig:bern-hayes} provides a schematic view of the crease patterns produced by these two reductions.

\begin{figure}[t]
\includegraphics[width=\columnwidth]{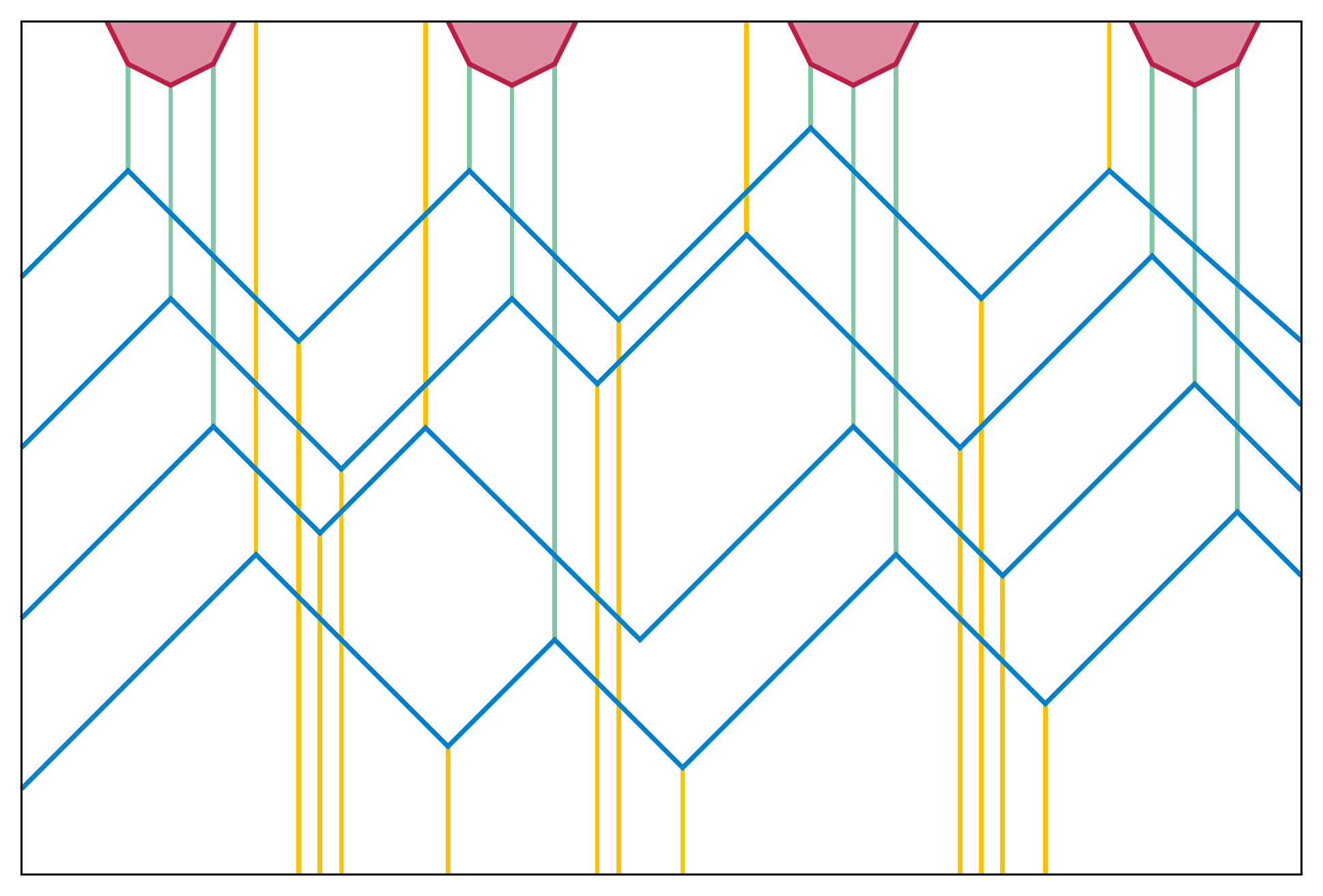}
\caption{Schematic view of the crease patterns produced by the hardness reductions of Bern and Hayes~\cite{BerHay-SODA-96}. The red regions at top are clause gadgets and the blue zigzag paths from left to right are variable gadgets. The variable gadgets are connected to the clause gadgets by vertical creases (light green) and additional ``noise'' vertical creases (yellow) connect to bends (``reflector gadgets'') in the paths of the variable gadgets. Not shown: the additional reflectors needed to complement variables. Illustration modeled after Fig. 10 of Bern and Hayes.}
\label{fig:bern-hayes}
\end{figure}

\begin{observation}
The local flat foldings of the crease patterns of Bern and Hayes have ply $O(1)$. For a NAE3SAT instance with $n$ vertices and $m$ clauses, they have treewidth $O(n)$, obtained by a path decomposition whose bags are the subsets of cells of the local flat folding intersected by vertical lines, in left-to-right order.
\end{observation}

\begin{theorem}
If the exponential time hypothesis is true, it is not possible to test flat foldability of crease patterns of ply $O(1)$ and treewidth $w$ in time $2^{o(w)}$, regardless of whether the pattern is labeled with mountain and valley folds or unlabeled.
\end{theorem}

\begin{proof}
If such a fast test existed, then applying it to the crease patterns produced by the hardness reductions of Bern and Hayes would give an algorithm for NAE3SAT with time $2^{o(m)}$, contradicting the exponential time hypothesis.
\end{proof}

\end{document}